\documentclass[12pt]{article}

\usepackage[hmargin=1.3in,vmargin=1.3in]{geometry}
\usepackage{bbding}
\usepackage{mathrsfs}
\usepackage{}
\usepackage{amsfonts}
\usepackage{multirow}
\usepackage{amsfonts,amssymb,amsmath,amsthm,bm}

\newtheorem{theorem}{Theorem}

\newtheorem{remark}{Remark}

\newtheorem{lemma}{Lemma}
\newtheorem{proposition}{Proposition}

\begin{document}

\title{Two new classes of quantum MDS codes}
\author{\small  Weijun Fang\thanks{Corresponding Author} $^{1}$  \   \ Fang-Wei Fu$^{1}$ \\
\small $^1$  Chern Institute of Mathematics and LPMC, Nankai University, Tianjin 300071, China\\
\small Email: nankaifwj@163.com, fwfu@nankai.edu.cn\\
}
\date{}
\maketitle
\thispagestyle{empty}
\begin{abstract}
Let $p$ be a prime and let $q$ be a power of $p$. In this paper, by using generalized Reed-Solomon (GRS for short) codes  and extended GRS codes, we construct two new classes of quantum maximum-distance-
separable (MDS) codes with parameters
\[ [[tq, tq-2d+2, d]]_{q} \]
for any $1 \leq t \leq q, 2 \leq d \leq \lfloor \frac{tq+q-1}{q+1}\rfloor+1$, and
\[ [[t(q+1)+2, t(q+1)-2d+4, d]]_{q} \]
for any $1 \leq t \leq q-1, 2 \leq d \leq t+2$ with $(p,t,d) \neq (2, q-1, q)$. Our quantum MDS codes have flexible parameters, and have minimum distances larger than $\frac{q}{2}+1$ when $t > \frac{q}{2}$. Furthermore, it turns out that our constructions generalize and improve some previous results.
\end{abstract}


\maketitle

%-------------------------------------------------------------------------------
\section{Introduction}

Quantum information and quantum computation have become a hot topic in recent decades. Quantum error-correcting codes are useful and have many applications in quantum computations and quantum communications. In \cite{4}, Calderbank et al. presented an effective mathematical method to construct nice quantum codes
from classical self-orthogonal codes over $\mathbb{F}_{2}$ or $\mathbb{F}_{4}$. Ashikhmin et al. then generalized this to the nonbinary case in \cite{2,5}. Since then, many good quantum codes have
been constructed by classical linear codes with certain self-orthogonality (see \cite{1,3,12,18,20,22}).

Let $q$ be a prime power. A $q$-ary quantum code $Q$ of
length $n$ is just defined as a subspace of the Hibert space $(\mathbb{C}^{q})^{\bigotimes n}$. The size of $Q$ is $K=\textnormal{dim}_{\mathbb{C}}Q$. We use $[[n, k, d]]_{q}$ to denote a $q$-ary quantum code of length
$n$ with size $q^{k}$ and minimum distance $d$. As in classical coding theory, one central theme in quantum
error-correction is the construction of codes that have reasonable good parameters. Similar to
the classical Singleton bound, the parameters of an $[[n, k, d]]_{q}$ quantum code have to satisfy the quantum Singleton bound:
$k \leq n - 2d + 2.$ (see \cite{5}). A quantum code achieving this quantum Singleton
bound is called a quantum maximum-distance-separable (MDS) code.

In the past decades, there have been a lot of research on the construction of
quantum MDS codes. It was proved in \cite{24} that the length of a
$q$-ary quantum stabilizer MDS code is at most $q^{2} + 1$ if the classical MDS conjecture holds. The problem of constructing q-ary quantum
MDS codes with $n \leq q + 1$ has been completely solved in \cite{18,20}. However, it is not an easy task to construct
quantum MDS codes with length $n > q +1$ and minimum distance $d > \frac{q}{2}+1$. Researchers
have made a great effort to construct such quantum MDS codes via  negacyclic codes (see \cite{16}), constacyclic codes (see \cite{10, 13, 15, 17, 21}), Pseudo-cyclic Codes (see \cite{19}) and generalized Reed-Solomon codes (see \cite{6,7,8,9,11,23}). We only list the known results for the constructions of $q$-ary quantum MDS codes with length $n=q^{2}$ or $q^{2}+1$ in Table 1.
\begin{table}[htbp]
\centering
\caption{Known Results of $q$-ary Quantum MDS Codes}

\vspace{3mm}

\begin{tabular}{|c|c|c|}
  \hline
  Length $n$ & Minimum distance $d$ & Reference\\
  \hline
  $n=q^{2}$ & $d \leq q$  & \cite{8}\\
  \hline
  \multirow{3}{*}{$n=q^{2}+1$} & $d=q+1$  &\cite{25}\\ \cline{2-3}
   & $d \leq q+1$, $q$ is even and $d$ is odd  &\cite{3}\\ \cline{2-3}
   & $d \leq q+1$, $q \equiv 1\textnormal{ (mod 4)}$ and $d$ is even &\cite{16} \\
  \hline
\end{tabular}
\end{table}

In this paper, by using the classical Hermitian self-orthogonal GRS codes and extended GRS codes, we present two new classes of quantum MDS codes. These quantum MDS codes have relatively large minimum distance in some certain range. More precisely, we construct quantum MDS codes with following parameters:
\begin{itemize}
  \item (1) $[[tq, tq-2d+2, d]]_{q}$, for any $1 \leq t \leq q$, $2 \leq d \leq \lfloor \frac{tq+q-1}{q+1}\rfloor+1$;
  \item (2) $[[t(q+1)+2, t(q+1)-2d+4, d]]_{q}$, for any $1 \leq t \leq q-1, 2 \leq d \leq t+2$ with $(p,t,d) \neq (2, q-1, q)$.
\end{itemize}
It is easy to see that the minimum distance of our quantum MDS codes can be larger than $\frac{q}{2}+1$ when $t > \frac{q}{2}$. Moreover, taking $t=q$ in (1), then $n=q^{2}$ and $d \leq q$, thus we obtain one of the main results in \cite{8}; Taking $t=q-1$ in (2), then $n=q^{2}+1$ and $d \leq q+1$ ($d \neq q$ if $q$ is even), which improves the results in \cite{3,16,25}.

The remainder of the paper is organized as follows. In Section 2, We state some basic notations
of quantum codes and Hermitian construction. We review and give some results about GRS codes and extended GRS codes in Section 3. Two new classes of quantum MDS codes are construted in Section 4. Finally, in Section 5 we use some conclusions to end this paper.

\section{Hermitian Construction}
Let $q$ be a prime power and let $\mathbb{F}_{q^{2}}$ be the finite field with $q^{2}$ elements. For any two vectors $\textbf{u}=(u_{1}, u_{2}, \ldots, u_{n}), \textbf{v}=(v_{1}, v_{2}, \ldots, v_{n}) \in \mathbb{F}_{q^{2}}^{n},$ we define their Euclidean and Hermitian inner product as
\[ \langle \textbf{u}, \textbf{v}\rangle :=\sum_{i=1}^{n}u_{i}v_{i} \]
and
\[ \langle \textbf{u}, \textbf{v}\rangle_{H} :=\sum_{i=1}^{n}u_{i}v_{i}^{q}, \]
respectively.

For an $\mathbb{F}_{q^{2}}$-linear code $\mathcal{C}$ of length $n$ and dimension $k$, the Euclidean and Hermitian dual code of $\mathcal{C}$ are defined as
\[ \mathcal{C}^{\perp} :=\{\textbf{u} \in \mathbb{F}_{q^{2}}^{n} :\langle \textbf{u}, \textbf{v}\rangle=0 \textnormal{ for all } \textbf{v} \in \mathcal{C} \} \]
and
\[ \mathcal{C}^{\perp_{H}} :=\{\textbf{u} \in \mathbb{F}_{q^{2}}^{n} :\langle \textbf{u}, \textbf{v}\rangle_{H}=0 \textnormal{ for all } \textbf{v} \in \mathcal{C} \},\]
respectively.

For any vector $\textbf{u}=(u_{1}, u_{2}, \ldots, u_{n}) \in \mathbb{F}_{q^{2}}^{n},$ we denote by $\textbf{u}^{q}$ the vector $(u_{1}^{q}, u_{2}^{q}, \ldots, u_{n}^{q})$. Let $G$ be the generator matrix of the code $\mathcal{C}$. Then it is not hard to prove that a vertor $\textbf{u}$ belongs to $\mathcal{C}^{\perp_{H}}$ if and only if $\textbf{u}^{q}\cdot G^{T}=0$, where $G^{T}$ is the transpose of $G$. An $\mathbb{F}_{q^{2}}$-linear code $\mathcal{C}$ of length $n$ is called Hermitian self-orthogonal if $\mathcal{C} \subseteq \mathcal{C}^{\perp_{H}}.$

Let $Q$ be a $q$-ary quantum MDS code with parameters $[[n, k, d]]$. Then it can detect up to $d - 1$ quantum errors and correct up to $\lfloor \frac{d-1}{2}\rfloor$ quantum errors. So it is desirable to keep minimum distance $d$ as large as possible for fixed length $n$ and $q$. As mentioned in Section 1, any $q$-ary $[[n, k, d]]$-quantum MDS code must satisfy the following quantum Singleton bound.
\begin{lemma}(Quantum Singleton Bound)
Let $Q$ be a $q$-ary $[[n, k, d]]$-quantum code, then $2d \leq n-k+2$.
\end{lemma}
A quantum code achieving this quantum Singleton bound is called a quantum MDS code.

In \cite{5}, Ashikhmin et al. presented the following Hermitian construction to construct quantum codes from classical codes.

\begin{lemma}(Hermitian Construction)
If there exists a $q^{2}$-ary $[n, k, d]$-linear code $\mathcal{C}$ such that $\mathcal{C^{\perp_{H}}} \subseteq \mathcal{C}$, then there exists a $q$-ary $[[n, 2k-n, \geq d]]$-quantum code.
\end{lemma}
Since the Hermitian dual code of an MDS code is also an MDS code, we obtain the following Hermitian construction for quantum MDS codes from Lemma 2.
\begin{lemma}(Hermitian Construction for Quantum MDS Codes)
If there exists a $q^{2}$-ary $[n, k, n-k+1]$-Hermitian self-orthogonal MDS code, then there exists a $q$-ary $[[n, n-2k, k+1]]$-quantum MDS code.
\end{lemma}

\section{GRS Codes and Extended GRS Codes}
In this section, we introduce some basic notions and properties of GRS codes and extended GRS codes. Let $a_{1}, a_{2}, \ldots, a_{n}$ be $n$ distinct elements of $\mathbb{F}_{q^{2}}$ and $v_{1}, v_{2}, \ldots, v_{n}$ be $n$ nonzero elements of $\mathbb{F}_{q^{2}}$. Put $\textbf{a}= (a_{1}, \ldots, a_{n})$ and $\textbf{v}=(v_{1}, \ldots, v_{n})$. Then the generalized Reed-Solomon (GRS for short) code over $\mathbb{F}_{q^{2}}$ associated to $\textbf{a}$ and $\textbf{v}$ is defined as follows:
$$GRS_{k}(\textbf{a}, \textbf{v}) := \{(v_{1}f(a_{1}), \ldots, v_{n}f(a_{n})) : f(x) \in \mathbb{F}_{q^{2}}[x], \deg(f(x)) \leq k-1 \}.$$
It is well known that the code $GRS_{k}(\textbf{a}, \textbf{v})$ is an $[n, k, n - k + 1]$-MDS code. A generator matrix of $GRS_{k}(\textbf{a}, \textbf{v})$ is given by
\begin{equation}\label{1}
  G_{k}(\textbf{a}, \textbf{v})=\left(
      \begin{array}{cccc}
        v_{1} & v_{2} & \cdots & v_{n} \\
        v_{1}a_{1} & v_{2}a_{2} & \cdots & v_{n}a_{n} \\
        \vdots & \vdots & \ddots & \vdots \\
        v_{1}a_{1}^{k-1} & v_{2}a_{2}^{k-1} & \cdots & v_{n}a_{n}^{k-1} \\
      \end{array}
    \right).
\end{equation}
Furthermore we consider the extended code of $GRS_{k}(\textbf{a}, \textbf{v})$ given by
\begin{eqnarray*}
  GRS_{k}(\textbf{a}, \textbf{v},\infty)&:=& \{(v_{1}f(a_{1}), \ldots, v_{n}f(a_{n}),f_{k-1}) \\
    & & :f(x) \in \mathbb{F}_{q^{2}}[x], \deg(f(x)) \leq k-1 \},
\end{eqnarray*}
where $f_{k-1}$ stands for the coefficient of $x^{k-1}$ in $f(x)$. It is easy to show that $GRS_{k}(\textbf{a}, \textbf{v},\infty)$ is an $[n+1, k, n - k + 2]$-MDS code (see \cite[Theorem 5.3.4]{26}). A generator matrix of $GRS_{k}(\textbf{a}, \textbf{v},\infty)$ is given by
\begin{equation}\label{1}
  G_{k}(\textbf{a}, \textbf{v}, \infty)=\left(
      \begin{array}{ccccc}
        v_{1} & v_{2} & \cdots & v_{n} & 0 \\
        v_{1}a_{1} & v_{2}a_{2} & \cdots & v_{n}a_{n} & 0\\
        \vdots & \vdots & \ddots & \vdots & \vdots\\
        v_{1}a_{1}^{k-2} & v_{2}a_{2}^{k-2} & \cdots & v_{n}a_{n}^{k-2} & 0 \\
        v_{1}a_{1}^{k-1} & v_{2}a_{2}^{k-1} & \cdots & v_{n}a_{n}^{k-1} & 1\\
      \end{array}
    \right).
\end{equation}
Denote by $\textbf{1}$ the all one vector $(1, 1, \ldots, 1)$. Throughout this paper, we always denote
\[w_{i}=\prod_{1 \leq j \leq n, j \neq i}(a_{i}-a_{j})^{-1}\textnormal{ for } 1 \leq i \leq n. \]
Then we have the following lemmas for the Euclidean duals of the codes $GRS_{k}(\textbf{a}, \textbf{1})$ and $GRS_{k}(\textbf{a}, \textbf{1}, \infty)$.

\begin{lemma}(see \cite{14})
\[ GRS_{k}(\textbf{a}, \textbf{1})^{\perp} =\{(w_{1}g(a_{1}), \ldots, w_{n}g(a_{n})) : g(x) \in \mathbb{F}_{q^{2}}[x], \deg(g(x)) \leq n-k-1 \}. \]
\end{lemma}

\vskip 3mm

\begin{lemma}
\begin{eqnarray*}
  GRS_{k}(\textbf{a}, \textbf{1}, \infty)^{\perp} &=& \{(w_{1}g(a_{1}), \ldots, w_{n}g(a_{n}), -g_{n-k}): \\
   & & g(x) \in \mathbb{F}_{q^{2}}[x], \deg(g(x)) \leq n-k \},
\end{eqnarray*}
where $g_{n-k}$ stands for the coefficient of $x^{n-k}$ in $g(x)$.
\end{lemma}
\begin{proof}
  Both sides are $\mathbb{F}_{q^{2}}$-linear space of dimension $n-k+1$. So we only need to show that the right hand is contained in the left hand. Let $ g(x) \in \mathbb{F}_{q^{2}}[x]$ with $\deg(g(x)) \leq n-k$ and $\textbf{u}=(w_{1}g(a_{1}), \ldots, w_{n}g(a_{n}), -g_{n-k})$. For $0 \leq s \leq k-1$, let $f_{s}(x)=g(x)x^{s}$. By the Lagrange interpolation formula,
  \[f_{s}(x)=\sum_{i=1}^{n}\prod_{1 \leq j \leq n, j\neq i}\frac{x-a_{j}}{a_{i}-a_{j}}f_{s}(a_{i})=\sum_{i=1}^{n}\prod_{1 \leq j \leq n, j\neq i}(x-a_{j})w_{i}g(a_{i})a_{i}^{s}. \]
  Comparing the coefficients of $x^{n-1}$ of the both sides, we have
\[ \sum_{i=1}^{n}w_{i}g(a_{i})a_{i}^{s}=0 \textnormal{ for all } 0 \leq s \leq k-2, \]
and
\[ \sum_{i=1}^{n}w_{i}g(a_{i})a_{i}^{k-1}=g_{n-k}. \]
We deduce that $\textbf{u} \cdot  G_{k}'(\textbf{a}, \textbf{1})^{T}=0$, i.e., $\textbf{u} \in GRS_{k}(\textbf{a}, \textbf{1}, \infty)^{\perp}$. The lemma follows.
\end{proof}

The following lemma presents a sufficient and necessary condition under which a codeword $\textbf{c}$ of $GRS_{k}(\textbf{a}, \textbf{v})$ is contained in $GRS_{k}(\textbf{a}, \textbf{v})^{\perp_{H}}$.

\begin{lemma}
A codeword $\textbf{c}=(v_{1}f(a_{1}), v_{2}f(a_{2}), \ldots, v_{n}f(a_{n}))$ of $GRS_{k}(\textbf{a}, \textbf{v})$ is contained in $GRS_{k}(\textbf{a}, \textbf{v})^{\perp_{H}}$ if and only if there exists a polynomial $g(x)$ with $\deg(g(x)) \leq n-k-1$, such that
\[ (v_{1}^{q+1}f^{q}(a_{1}), v_{2}^{q+1}f^{q}(a_{2}), \ldots, v_{n}^{q+1}f^{q}(a_{n}))=(w_{1}g(a_{1}), w_{2}g(a_{2}), \ldots, w_{n}g(a_{n})). \]
\end{lemma}
\begin{proof}
By Eq. (1), we may write the generator matrix of $GRS_{k}(\textbf{a}, \textbf{v})$ as
\[ G_{k}(\textbf{a}, \textbf{v})=G_{k}(\textbf{a}, \textbf{1})\cdot D, \]
where $D$ is the diagonal matrix $diag\{v_{1}, v_{2}, \ldots, v_{n} \}$. Thus
\begin{eqnarray*}
  \textbf{c} \in GRS_{k}(\textbf{a}, \textbf{v})^{\perp_{H}} &\Leftrightarrow& \textbf{c}^{q}\cdot G_{k}(\textbf{a}, \textbf{v})^{T}=\textbf{c}^{q}\cdot D\cdot G_{k}(\textbf{a}, \textbf{1})^{T}=0 \\
    &\Leftrightarrow& \textbf{c}^{q}\cdot D \in GRS_{k}(\textbf{a}, \textbf{1})^{\perp}
\end{eqnarray*}
The desired result then follows from Lemma 4.
\end{proof}
\begin{remark}
Indeed, Lemma 6 is a generalization of \cite{14}, which consider the case of Euclidean dual.
\end{remark}

Similarly, we have the following result for the case of extended GRS code $GRS_{k}(\textbf{a}, \textbf{v}, \infty).$
\begin{lemma}
A codeword $\textbf{c}=(v_{1}f(a_{1}), v_{2}f(a_{2}), \ldots, v_{n}f(a_{n}), f_{k-1})$ of $GRS_{k}(\textbf{a}, \textbf{v}, \infty)$ is contained in $GRS_{k}(\textbf{a}, \textbf{v}, \infty)^{\perp_{H}}$ if and only if there exists a polynomial $g(x)$ with $\deg(g(x)) \leq n-k$, such that
\begin{eqnarray*}
  (v_{1}^{q+1}f^{q}(a_{1}), v_{2}^{q+1}f^{q}(a_{2}), \ldots, v_{n}^{q+1}f^{q}(a_{n}), f^{q}_{k-1}) &=& (w_{1}g(a_{1}), w_{2}g(a_{2}),  \\
    &&\ldots, w_{n}g(a_{n}), -g_{n-k}).
\end{eqnarray*}
\end{lemma}
\begin{proof}
By (2), we may write the generator matrix of $GRS_{k}(\textbf{a}, \textbf{v})$ as
\[ G_{k}(\textbf{a}, \textbf{v}, \infty)=G_{k}(\textbf{a}, \textbf{1}, \infty)\cdot D', \]
where $D'$ is the diagonal matrix $diag\{v_{1}, v_{2}, \ldots, v_{n}, 1 \}$. Thus
\begin{eqnarray*}
  \textbf{c} \in GRS_{k}(\textbf{a}, \textbf{v}, \infty)^{\perp_{H}} &\Leftrightarrow& \textbf{c}^{q}\cdot G_{k}(\textbf{a}, \textbf{v}, \infty)^{T}=\textbf{c}^{q}\cdot D' \cdot G_{k}(\textbf{a}, \textbf{1}, \infty)^{T}=0 \\
    &\Leftrightarrow& \textbf{c}^{q}\cdot D' \in GRS_{k}(\textbf{a}, \textbf{1}, \infty)^{\perp}
\end{eqnarray*}
The desired result then follows from Lemma 5.
\end{proof}

\section{New Quantum MDS Codes}
In this section, we construct quantum MDS codes by using the GRS codes of length $tq$ and the extended GRS codes of length $t(q+1)+2$, respectively.
\subsection{Length $n=tq$ with $1 \leq t \leq q$}
Label the elements of the set $\mathbb{F}_{q}$ as $\beta_{1}, \beta_{2}, \ldots, \beta_{q}$. Fix an $\alpha \in \mathbb{F}_{q^{2}}\setminus \mathbb{F}_{q}$. Denote $F_{i}=\mathbb{F}_{q}+\beta_{i}\alpha :=\{x+\beta_{i}\alpha : x \in \mathbb{F}_{q}\}.$ Then $F_{i} \bigcap F_{j} = \emptyset$ for any $1 \leq i \neq j \leq q$. Let $1 \leq t \leq q$ and $n=tq$. Denote $X=\bigcup_{i=1}^{t}F_{i}$. Label the elements of the set $X$ as $a_{1}, a_{2}, \ldots, a_{n}$.

\begin{lemma}
Keep the above notations, then
\begin{description}
  \item[(1)] For any $\tau \in \mathbb{F}_{q}$, we have $$\prod_{h \in \mathbb{F}_{q}}(\tau\alpha-h)=\tau(\alpha^{q}-\alpha);$$
  \item[(2)] For any $1 \leq s \leq t$ and $b \in F_{s}$, we have $$\prod_{h \in F_{s}, h \neq b}(b-h)=(-1)^{q};$$
  \item[(3)] For any $1 \leq s \neq j \leq t$ and $b \in F_{s}$, we have $$\prod_{h \in F_{j}}(b-h)=(\beta_{s}-\beta_{j})(\alpha^{q}-\alpha).$$
\end{description}
\end{lemma}
\begin{proof}
\begin{description}
  \item[(1)] $\prod_{h \in \mathbb{F}_{q}}(\tau\alpha-h)=\tau^{q}\alpha^{q}-\tau\alpha=\tau(\alpha^{q}-\alpha);$
  \item[(2)] Suppose that $b=\theta+\beta_{s}\alpha \in F_{s}$ for some $\theta \in \mathbb{F}_{q}$. Then \[\prod_{h \in F_{s}, h \neq b}(b-h)=\prod_{\gamma \in \mathbb{F}_{q}, \gamma \neq \theta}(\theta+\beta_{s}\alpha-(\gamma+\beta_{s}\alpha))=\prod_{h \in \mathbb{F}_{q}, h \neq 0} h=(-1)^{q}.\]
  \item[(3)] Suppose that $b=\theta+\beta_{}\alpha \in F_{s}$ for some $\theta \in \mathbb{F}_{q}$. Then
\begin{eqnarray*}
  \prod_{h \in F_{j}}(b-h) &=& \prod_{\gamma \in \mathbb{F}_{q}}(\theta+\beta_{s}\alpha-(\gamma+\beta_{j}\alpha)) \\
    &=& \prod_{h \in \mathbb{F}_{q}}((\beta_{s}-\beta_{j})\alpha-h) \\
    &=& (\beta_{s}-\beta_{j})(\alpha^{q}-\alpha).
\end{eqnarray*}
 The last equality follows by part (1).
\end{description}
\end{proof}
\begin{lemma}
For a given $i$, suppose $a_{i} \in F_{s}$ for some $1 \leq s \leq t$. Then we have
\[\prod_{1 \leq j \leq n, j\neq i}(a_{i}-a_{j})=(-1)^{q}(\alpha^{q}-\alpha)^{t-1}\prod_{1 \leq j \leq t, j\neq s}(\beta_{s}-\beta_{j}). \]
\end{lemma}
\begin{proof}
By Lemma 8, we have
\begin{eqnarray*}
  \prod_{1 \leq j \leq n, j\neq i}(a_{i}-a_{j}) &=& \prod_{\xi_{s} \in F_{s}, \xi_{s} \neq a_{i}}(a_{i}-\xi_{s})\prod_{1 \leq j \leq t, j\neq s}\prod_{\xi_{j} \in F_{j}}(a_{i}-\xi_{j}) \\
   &=&(-1)^{q} \cdot  \prod_{1 \leq j \leq t, j\neq s}(\alpha^{q}-\alpha)(\beta_{s}-\beta_{j}) \\
   &=&(-1)^{q}(\alpha^{q}-\alpha)^{t-1}\prod_{1 \leq j \leq t, j\neq s}(\beta_{s}-\beta_{j}).
\end{eqnarray*}
\end{proof}

 Using the aforementioned lemmas, we can provide our first construction of $q$-ary quantum MDS codes.
 \begin{theorem}
 Let $q$ be a prime power and $1 \leq t \leq q$. Then, there exists a $q$-ary $[[tq, tq-2d+2, d]]$-quantum MDS code, where $2 \leq d \leq \lfloor \frac{tq+q-1}{q+1}\rfloor+1$.
 \end{theorem}
\begin{proof}
Let $n=tq$ and $a_{1}, a_{2}, \ldots, a_{n}$ be defined as above. For a given $i$,  suppose $a_{i} \in F_{s}$ for some $1 \leq s \leq t$. Then by Lemma 9,
\[w_{i}=\prod_{1 \leq j \leq n, j\neq i}(a_{i}-a_{j})^{-1}=(-1)^{q}(\alpha^{q}-\alpha)^{1-t}\prod_{1 \leq j \leq t, j\neq s}(\beta_{s}-\beta_{j})^{-1}.\]
So $w_{i}(\alpha^{q}-\alpha)^{t-1}=(-1)^{q}\prod_{1 \leq j \leq t, j\neq s}(\beta_{s}-\beta_{j})^{-1} \in \mathbb{F}_{q}^{*}.$
Thus there exists $v_{i} \in \mathbb{F}_{q^{2}}^{*}$ such that $v_{i}^{q+1}=w_{i}(\alpha^{q}-\alpha)^{t-1}$. Denote $\textbf{a}=(a_{1}, a_{2}, \ldots, a_{n})$ and $\textbf{v}=(v_{1}, v_{2}, \ldots, v_{n})$. For $1 \leq k \leq \lfloor \frac{tq+q-1}{q+1}\rfloor$, we consider $GRS_{k}(\textbf{a}, \textbf{v})$. For any codeword $\textbf{c}=(v_{1}f(a_{1}), v_{2}f(a_{2}), \ldots, v_{n}f(a_{n}))$ with $\deg(f(x)) \leq k-1$. Let $g(x)=(\alpha^{q}-\alpha)^{t-1}f^{q}(x)$. Then it is easy to see that $\deg(g(x)) \leq q(k-1) \leq n-k-1$. Note that $v_{i}^{q+1}f^{q}(a_{i})=w_{i}(\alpha^{q}-\alpha)^{t-1}f^{q}(a_{i})=w_{i}g(a_{i})$. According to Lemma 6, $\textbf{c} \in GRS_{k}(\textbf{a}, \textbf{v})^{\perp_{H}}.$ Thus $GRS_{k}(\textbf{a}, \textbf{v})\subseteq GRS_{k}(\textbf{a}, \textbf{v})^{\perp_{H}},$ i.e., $GRS_{k}(\textbf{a}, \textbf{v})$ is a $q^{2}$-ary $[tq, k, tq-k+1]$-Hermitian self-orthogonal MDS code. The desired conclusion then follows from Lemma 3 ($d=k+1$).
\end{proof}
\begin{remark}
\begin{description}
  \item[(1)] Taking $t=q$ in Theorem 1, we obtain a family of $q$-ary quantum MDS code of length $q^{2}$ and minimum distance $d \leq q$. So one of the main results in \cite{8} (see Table 1) is a special case of Theorem 1.
  \item[(2)] In \cite[Example 4.8]{9}, the authors obtained a $q$-ary $[[tq, tq-2d+2, d]]$-quantum MDS code with $d \leq \frac{q}{2}+2$. Note that when $\frac{q}{2}+2 < t \leq q, \lfloor \frac{tq+q-1}{q+1}\rfloor+1 \geq t > \frac{q}{2}+2$, thus our code has larger minimum distance than theirs.
\end{description}

\end{remark}
\subsection{Length $n+1=t(q+1)+2$ with $1 \leq t \leq q-1$}
Let $1 \leq t \leq q-1$ and $n=t(q+1)+1$. Let $\theta \in \mathbb{F}_{q^{2}}$ be a $(q+1)$-th primitive root of unity and $\langle\theta \rangle$ be the cyclic group generated by $\theta$. Choose $\beta_{1}, \ldots , \beta_{t} \in \mathbb{F}_{q^{2}}^{*}$
such that $\{\beta_{s} \langle\theta \rangle\}^{t}_{s
=1}$ represent distinct cosets of $\mathbb{F}_{q^{2}}^{*}/\langle \theta \rangle$. Let $A_{s}=\beta_{s} \langle\theta \rangle$ and $A=(\bigcup^{t}_{s
=1}A_{s}) \bigcup \{0\}$. Label the elements of the set
$A$ by $a_{1}, a_{2}, \ldots , a_{n-1}, a_{n}=0$. We first calculate $w_{i}=\prod_{1 \leq j \leq n, j \neq i}(a_{i}-a_{j})^{-1}$, for all $1 \leq i \leq n$.

\begin{lemma}
\[ \prod_{i=1}^{n-1}(a_{n}-a_{i})=(-1)^{n-1+qt}\prod_{s=1}^{t}\beta_{s}^{q+1}. \]
In particular, $\prod_{i=1}^{n-1}(a_{n}-a_{i}) \in \mathbb{F}_{q}$.
\end{lemma}

\begin{proof}
\begin{eqnarray*}
  \prod_{i=1}^{n-1}(a_{n}-a_{i}) &=& (-1)^{n-1}\prod_{i=1}^{n-1}a_{i}=(-1)^{n-1}\prod_{s=1}^{t}\prod_{x_{s} \in A_{s}}x_{s} \\
    &=& (-1)^{n-1}\prod_{s=1}^{t}(\theta^{\frac{q(q+1)}{2}}\beta_{s}^{q+1})=(-1)^{n-1+qt}\prod_{s=1}^{t}\beta_{s}^{q+1}.
\end{eqnarray*}
Note that each $\beta_{s}^{q+1} \in \mathbb{F}_{q}$. The lemma follows.
\end{proof}

Note that
\begin{equation}\label{3}
  \prod_{0 \leq \ell \leq q}(x-\theta^{\ell})=x^{q+1}-1
\end{equation}
and for $0 \leq m \leq q$
\[ \prod_{0 \leq \ell \leq q, \ell \neq m}(x-\theta^{\ell})=\frac{x^{q+1}-1}{x-\theta^{m}}=\sum_{i=0}^{q}x^{i}\theta^{m(q-i)}. \]
So
\begin{equation}\label{4}
  \prod_{0 \leq \ell \leq q, \ell \neq m}(\theta^{m}-\theta^{\ell})=\sum_{i=0}^{q}(\theta^{m})^{i}\theta^{m(q-i)}=(q+1)\theta^{qm}=\theta^{qm}.
\end{equation}
\begin{lemma}
Given $1 \leq i \leq n-1$. Suppose that $a_{i} \in A_{r}$, i.e., $a_{i}=\beta_{r}\theta^{m}$ for some $1 \leq r \leq t$ and $0 \leq m \leq q$. Then,
\[ \prod_{1 \leq j \leq n, j \neq i}(a_{i}-a_{j})=a_{i}^{q+1}\prod_{1 \leq s \leq t, s \neq r}(\beta_{r}^{q+1}-\beta_{s}^{q+1}). \]
In particular, $\prod_{1 \leq j \leq n, j \neq i}(a_{i}-a_{j}) \in \mathbb{F}_{q}$.
\end{lemma}

\begin{proof}
\[\prod_{1 \leq j \leq n, j \neq i}(a_{i}-a_{j})=a_{i}\prod_{x_{r} \in A_{r}, x_{r} \neq a_{i}}(a_{i}-x_{r})\prod_{1 \leq s \leq t, s \neq r}\prod_{x_{s} \in A_{s}}(a_{i}-x_{s}). \]
From Eq. (4), we have
\[\prod_{x_{r} \in A_{r}, x_{r} \neq a_{i}}(a_{i}-x_{r})=\prod_{0 \leq \ell \leq q, \ell \neq m}(\beta_{r}\theta^{m}-\beta_{r}\theta^{\ell})=\beta_{r}^{q}\theta^{qm}=a_{i}^{q}. \]
From Eq. (3), we have
\[\prod_{x_{s} \in A_{s}}(a_{i}-x_{s})=\prod_{0 \leq \ell \leq q}(\beta_{r}\theta^{m}-\beta_{s}\theta^{\ell})=\beta_{s}^{q+1}((\frac{\beta_{r}\theta^{m}}{\beta_{s}})^{q+1}-1)=\beta_{r}^{q+1}-\beta_{s}^{q+1}. \]
The first conclusion then follows. Since $a_{i}^{q+1}, \beta_{s}^{q+1} \in \mathbb{F}_{q}$, the second conclusion is also proved.
\end{proof}

From Lemmas 10 and 11, we know that for all $1 \leq i \leq n,$
\[w_{i}=\prod_{1 \leq j \leq n, j \neq i}(a_{i}-a_{j})^{-1} \in \mathbb{F}_{q}.\]
We let $\gamma_{i} \in \mathbb{F}_{q^{2}}^{*}$ such that
\begin{equation}\label{5}
  \gamma_{i}^{q+1}=-w_{i}.
\end{equation}

Before giving our second construction, we need another simple lemma.
\begin{lemma}
For any integer $\ell \geq 2$ and finite field $F$, there exists a monic polynomial $m(x) \in F[x]$ of degree $\ell$ such that $m(a) \neq 0 $ for all $a \in F$.
\end{lemma}
\begin{proof}
Let $E$ be the extension field of $F$ with degree $\ell$, and $\alpha$ be a primitive element of $E$. Then the minimal polynomial $m(x)$ of $\alpha$ is  an irreducible polynomial over $F[x]$ of degree $\ell$. Since $\ell \geq 2$, $(x-a) \nmid m(x)$ for all $a \in F$, i.e., $m(a) \neq 0 $ for all $a \in F$.
\end{proof}

Now, we provide a construction of Hermitian self-orthogonal MDS codes by Lemma 7.
\begin{proposition}
Let $p$ be a prime and let $q$ be a power of $p$. Then,
\begin{description}
  \item[(1)] for any $1 \leq t \leq q-1$ and $1 \leq k \leq t+1$ with $(t,k) \neq (q-1, q-1)$, there exists a $q^{2}$-ary $[t(q+1)+2, k, t(q+1)+3-k]$-Hermitian self-orthogonal MDS code;
  \item[(2)] if $p \neq 2$, for $(t,k) = (q-1, q-1)$, there still exists a $q^{2}$-ary $[q^{2}+1, q-1, q^{2}-q+3]$-Hermitian self-orthogonal MDS code.
\end{description}
\end{proposition}
\begin{proof}
Let $n=t(q+1)+1$, and $a_{1}, a_{2}, \ldots, a_{n}$ be defined as the beginning of this subsection. Let $\ell = t+1-k$.

We first to prove Part (1) which consider the case $(t,k) \neq (q-1, q-1)$.
 \begin{description}
  \item[(i)] If $k \leq t-1$, then $\ell \geq 2$. From Lemma 12, there exists a monic polynomial $m(x) \in \mathbb{F}_{q^{2}}[x]$ of degree $\ell$, such that $m(a_{i}) \neq 0$ for all $1 \leq i \leq n$;
  \item[(ii)] If $k=t$, then $\ell=1$ and $t \neq q-1$ since $(t, k) \neq (q-1, q-1)$. Thus $n=t(q+1)+1 < q^{2}$ and we may choose an element $\alpha \in \mathbb{F}_{q^{2}} \backslash \{a_{1}, a_{2}, \cdots, a_{n} \}$. Let $m(x)=x-\alpha$, then $m(a_{i}) \neq 0$ for all $1 \leq i \leq n$;
  \item[(iii)] If $k=t+1$, then $\ell=0$. We let $m(x)=1$.
\end{description}
 In conclusion, there always exists a monic polynomial $m(x) \in \mathbb{F}_{q^{2}}[x]$ of degree $\ell$, such that $m(a_{i}) \neq 0$ for all $1 \leq i \leq n$. Let $v_{i}=m(a_{i})\gamma_{i}\neq 0$, where $\gamma_{i}$ is defined by Eq. (5). Put $\textbf{a}=(a_{1}, a_{2}, \cdots, a_{n})$ and $\textbf{v}=(v_{1}, v_{2}, \cdots, v_{n})$, we consider the extended GRS code $GRS_{k}(\textbf{a}, \textbf{v}, \infty)$.

Suppose that $\textbf{c}=(v_{1}f(a_{1}), v_{2}f(a_{2}), \ldots, v_{n}f(a_{n}), f_{k-1})$ is a codeword of $GRS_{k}(\textbf{a}, \textbf{v}, \infty)$, i.e., $\deg(f) \leq k-1$ and $f_{k-1}$ is the coefficient of $x^{k-1}$ in $f(x)$. Let $g(x)=-m^{q+1}(x)f^{q}(x)$. Then
\[v_{i}^{q+1}f^{q}(a_{i})=((m(a_{i})\gamma_{i})^{q+1}f^{q}(a_{i})=w_{i}g(a_{i}).\]
If $\deg(f) < k-1$, i.e., $f_{k-1}=0$, then $\deg(g)=\ell(q+1)+q\deg(f)=(t+1-k)(q+1)+q\deg(f)=n-k-(k-1-\deg(f))q < n-k$, thus $g_{n-k}=0=-f_{k-1}^{q}$; If $\deg(f) = k-1$, then $\deg(g)=n-k$, thus $g_{n-k}=-f_{k-1}^{q}$. Thus we always have
\[f_{k-1}^{q}=-g_{n-k}.\]
Then by Lemma 7, we have $\textbf{c} \in GRS_{k}(\textbf{a}, \textbf{v}, \infty)^{\perp_{H}}.$ Thus $GRS_{k}(\textbf{a}, \textbf{v}, \infty) \subseteq GRS_{k}(\textbf{a}, \textbf{v}, \infty)^{\perp_{H}},$ i.e., $GRS_{k}(\textbf{a}, \textbf{v}, \infty)$ is a $q^{2}$-ary $[t(q+1)+2, k, t(q+1)+3-k]$-Hermitian self-orthogonal MDS code. The Part (1) is proved.

Next, we suppose that $p \neq 2$ and $(t, k)=(q-1, q-1)$. Then $n=q^{2}$ and $\{a_{1}, a_{2}, \ldots, a_{n}\}=\mathbb{F}_{q^{2}}$. Choose an element $\pi \in \mathbb{F}_{q^{2}}\backslash \mathbb{F}_{q}$. Let $m(x)=x^{q}+x-\pi$. Since for any $\alpha \in \mathbb{F}_{q^{2}}$. $\alpha^{q}+\alpha \in \mathbb{F}_{q}$, we have
\[m(a_{i}) \neq 0 \textnormal{ for all }1 \leq i \leq n.\]
Note that

\begin{eqnarray}
 \nonumber m^{q+1}(a_{i}) &=& (a_{i}^{q}+a_{i}-\pi)^{q+1}=(a_{i}^{q}+a_{i}-\pi)^{q}(a_{i}^{q}+a_{i}-\pi) \\
    &=& a_{i}^{2q}+2a_{i}^{q+1}-(\pi+\pi^{q})(a_{i}^{q}+a_{i})+\pi^{q+1}.
\end{eqnarray}
 Now for any polynomial $f(x) \in  \mathbb{F}_{q^{2}}[x]$ with $\deg(f) \leq k-1=q-2$.
\begin{description}
  \item[(i)] If $\deg(f(x))=q-2$. Then $\deg(f^{q}(x))=q^{2}-2q.$ Note that $a_{i}^{2q}a_{i}^{q^{2}-2q}=a_{i}$. Then there exists a polynomial $g(x) \in \mathbb{F}_{q^{2}}[x]$ with $\deg(g)=(q^{2}-2q)+(q+1)=q^{2}-q+1=n-k$, such that $g(a_{i})=-\frac{1}{2}m^{q+1}(a_{i})f^{q}(a_{i})$  ($\frac{1}{2}$ is the element in $\mathbb{F}_{q}$ such that $2\cdot \frac{1}{2}=1$) and $g_{n-k}=-f^{q}_{k-1};$
  \item[(ii)] If $\deg(f) \leq q-3$. Then there exists a polynomial $g(x) \in \mathbb{F}_{q^{2}}[x]$ with $\deg(g)=q\deg(f)+2q \leq q^{2}-q < n-k$ such that $g(a_{i})=-\frac{1}{2}m^{q+1}(a_{i})f^{q}(a_{i})$ and $g_{n-k}=-f^{q}_{k-1}=0.$
\end{description}

Let $\theta \in \mathbb{F}_{q^{2}}$ such that $\theta^{q+1}=\frac{1}{2}$, and let $v_{i}=\theta m(a_{i})\gamma_{i}$, where $\gamma_{i}$ is defined by Eq. (5). Then from the above discussion, for any codeword of $GRS_{k}(\textbf{a}, \textbf{v}, \infty)$, $\textbf{c}=(v_{1}f(a_{1}), v_{2}f(a_{2}), \ldots, v_{n}f(a_{n}), f_{k-1})$, where $\deg(f) \leq k-1=q-2$, there always there exists a polynomial  $g(x) \in \mathbb{F}_{q^{2}}[x]$ with $\deg(g) \leq n-k$, such that
\[w_{i}g(a_{i})=-\gamma_{i}^{q+1}(-\frac{1}{2}m^{q+1}(a_{i})f^{q}(a_{i}))=v_{i}^{q+1}f^{q}(a_{i})\]
and
\[g_{n-k}=-f^{q}_{k-1}.\]
Then by Lemma 7, we have $\textbf{c} \in GRS_{k}(\textbf{a}, \textbf{v}, \infty)^{\perp_{H}}.$ Thus $GRS_{k}(\textbf{a}, \textbf{v}, \infty) \subseteq GRS_{k}(\textbf{a}, \textbf{v}, \infty)^{\perp_{H}},$ i.e., $GRS_{k}(\textbf{a}, \textbf{v}, \infty)$ is a $q^{2}$-ary $[q^{2}+1, q-1, q^{2}-q+3]$-Hermitian self-orthogonal MDS code. The Part (2) is also proved.

Our proposition follows.
\end{proof}

\begin{remark}
\begin{description}
  \item[(1)] When $(t, k) = (q-1, q-1)$, we have $\{a_{1}, a_{2}, \ldots, a_{n}\}=\mathbb{F}_{q^{2}}$ and $\ell=t-k+1=1$. Then for any $m(x) \in \mathbb{F}_{q^{2}}[x]$ of degree 1, there always exists some $1 \leq i \leq n$ such that $m(a_{i})=0$. At this moment, $v_{i}=m(a_{i})\gamma_{i}=0$. Thus the method in the proof of Part (1) is not available for the Part (2).
  \item[(2)] When $(p, t, k) = (2, q-1, q-1)$, the term $2a_{i}^{q+1}$ in Eq. (6) will vanishes in $\mathbb{F}_{q}$. Thus the method in the proof of Part (2) is not available for this case.
\end{description}

\end{remark}

We give the second construction in the following theorem, which immediately follows by Proposition 1 and Lemma 3.
\begin{theorem}
Let $p$ be a prime and let $q$ be a power of $p$. Then, for any $1 \leq t \leq q-1$ and $2 \leq d \leq t+2$ with $(p, t, d) \neq (2, q-1, q)$, there exists a $q$-ary $[[t(q+1)+2, t(q+1)-2d+4, d]]$-quantum MDS code.
\end{theorem}

\begin{remark}
\begin{description}
  \item[(1)] When $t > \frac{q}{2}-1$, the minimum distance of the quantum code of Theorem 2 may larger than $\frac{q}{2}+1$.
  \item[(2)] Taking $t=q-1$ in Theorem 2, we obtain a family of $q$-ary quantum MDS code of length $q^{2}+1$ and minimum distance $d \leq q+1$ ($d \neq q$ when $q$ is even). Comparing with the results in the references \cite{3,16,25}, our results improve theirs (see Table 1).
\end{description}
\end{remark}

Suppose that $C$ is a Hermitian self-orthogonal $[5, 1, 5]$-MDS code over $\mathbb{F}_{4}$. Let $\textbf{c}=(c_{0}, c_{1}, c_{2}, c_{2}, c_{4})$ be a nonzero codeword of $C$. Since the minimum distance of $C$ is 5 and $C$ is Hermitian self-orthogonal, we have $c_{i} \neq 0$ and $\sum_{i=0}^{4}c_{i}^{3}=0$. Note that $c_{i}^{3}=1$, we obtain that $5=0$, i.e., $1=0$ in $\mathbb{F}_{4}$. This is a contradiction. Thus there is no Hermitian self-orthogonal $[5, 1, 5]$-MDS code over $\mathbb{F}_{4}$. So we can not construct a 2-ary $[[5, 3, 2]]$-quantum MDS code via the Hermitian Construction (Lemma 3). One natural question is
\begin{itemize}
  \item \textbf{Problem} If $q$ is even and larger than 2, does there exist a $q$-ary quantum MDS code of length $q^{2}+1$ and minimum distance $q?$
\end{itemize}

\section{Conclusions}
In this paper, we firstly study some dual properties of classical GRS codes and extended GRS codes. Then, we apply the additive and multiplicative subgroups of $\mathbb{F}_{q^{2}}$ and their cosets to obtain Hermitian self-orthogonal MDS codes. Finally, two new families of quantum MDS codes with flexible parameters are presented from the Hermitian construction. Furthermore, we see that some previous results are special cases of ours and the parameters of some results are also improved. However, when $q$ is even and larger than 2, our method is not available for constructing quantum MDS code of length $q^{2}+1$ and minimum distance $q$. So we expect that this case and more quantum MDS codes with new parameters can be constructed from some other mathematical tools in the future work.

\vskip 3mm \noindent {\bf Acknowledgments} This research is supported by the National Key Basic Research Program of China (Grant No. 2013CB834204), and the National Natural Science Foundation of China (No. 61571243).

%-----------------------------------------------------------

\end{document}